\documentclass[a4paper,hidelinks, 12pt]{article}
\usepackage[utf8]{inputenc}
\usepackage{a4wide}
\usepackage{amsmath}
\usepackage{amsfonts}
\usepackage{amssymb}
\usepackage{amsthm}
\usepackage{blkarray}
\usepackage{subcaption}
\usepackage{cite}
\usepackage{xcolor}
\usepackage[toc,page]{appendix}
\usepackage{mathtools}
\usepackage{hyperref}
\usepackage{enumitem}
\usepackage{algorithm}

\numberwithin{equation}{section}
\allowdisplaybreaks

\newtheorem{theorem}{Theorem}
\newtheorem{prop}{Proposition}

\newtheorem*{con-non}{Conjecture}
\theoremstyle{plain}

\newcommand{\Tr}{{\rm tr}}

\newcommand{\Su}{{\mathfrak{su}}}
\newcommand{\So}{{\mathfrak{so}}}
\newcommand{\ad}{{\text{ad}}}
\newcommand{\SU}{{\mathsf{SU}}}

\begin{document}
\begin{titlepage}
\begin{center}

{\large \bf {Computable conditions for order-2  $CP$ symmetry in NHDM potentials}}

\vskip 1cm

R.~Plantey\footnote{E-mail: Robin.Plantey@ntnu.no} 
 and
M.~Aa.~Solberg\footnote{E-mail: Marius.Solberg@ntnu.no} 

\vspace{1.0cm}

Department of Structural Engineering, \\
Norwegian University of Science and Technology, \\
7034 Trondheim, Norway. \\

\end{center}

\vskip 3cm

\begin{abstract}
We derive necessary and sufficient conditions for order-2 $CP$ ($CP2$) symmetry in $N$-Higgs-doublet potentials for $N>2$. The conditions, which are formulated as relations between vectors that transform under the adjoint representation of $\mathsf{SU}(N)$ under a change of doublet basis, are representation theoretical in nature. Making use of Lie algebra and representation theory we devise an efficient, computable algorithm which may be applied to decide whether or not a given numerical potential is $CP2$ invariant. 
\end{abstract}

\end{titlepage}

\tableofcontents

\setcounter{footnote}{0}

\pagebreak
\section{Introduction}
\label{sect:intro}

The possibility for $CP$-violation is arguably one of the most attractive features of Multi-Higgs-Doublets models (NHDM), enabling them to accommodate baryogenesis~\cite{Sakharov:1967dj,Turok:1990zg} and contributing to much of their rich phenomenology~\cite{Cordero-Cid2016,Logan:2020mdz}.
Yet, as with other symmetries of NHDMs, $CP$ symmetry can be apparent in one doublet basis and completely obfuscated in another. Indeed, any basis transformation followed by the canonical $CP$ transformation $\Phi_i (\vec{x},t)\to \Phi_i^\ast (-\vec{x},t)$ can be a valid $CP$ symmetry. That is, a general $CP$ transformation takes the form
\begin{align}\label{E:genCP}
	CP: \Phi_i (\vec{x},t)&\to V_{ij}\Phi_j^\ast (-\vec{x},t),
\end{align}
for some matrix $V\in \mathsf{U}(N)$~\cite{Feinberg1959}.
In addition, a $CP$ symmetry need not be of order 2, with $CP^2$ being the identity, but may be of higher order $p=2q$, with $p>2$ the smallest integer such that $CP^p$ is the identity instead. In contrast to higher-order $CP$ symmetries, $CP2$ is equivalent to the existence of a basis where all the parameters are real~\cite{Gunion_2005}. In other words, $CP2$ is equivalent to the canonical $CP$ in some basis.
NHDMs with CPs of higher order than 2 often generate $CP2$ as an accidental symmetry, and $CP2$ was long considered to be the only $CP$. However, while in the 2HDM $CP2$ is the only possible $CP$, a 3HDM with an order-4 $CP$ symmetry ($CP4$) and no other symmetries was identified and studied in~\cite{Ivanov_2012, Ivanov_2016}. Even higher order $CP$s than $CP4$, with no accidental symmetries, were constructed and examined in~\cite{Ivanov_2018}.

Thus, establishing whether a particular potential breaks $CP$ is challenging but crucial for conducting a phenomenological analysis. For the general 2HDM, necessary and sufficient conditions for $CP$ symmetry were first derived in terms of basis-invariant quantities in~\cite{Gunion_2005} and later using basis-covariant quantities in the bilinear formalism~\cite{Ivanov2006360,Maniatis2008}. Methods based on basis-covariant objects proved to be quite powerful and have since been succesfully applied to the 3HDM to detect $CP2$, $CP4$ as well as other symmetries~\cite{Nishi:2006tg,Ivanov:2018ime,Nishi:2011gc,deMedeirosVarzielas:2019rrp}. In particular, within this framework, a complete solution for detecting $CP2$ for $N=3$ and a discussion of the cases $N>3$ was given in~\cite{Nishi:2006tg}. In this work, we show that this idea, formulated in the language of representation theory, can be extended to derive necessary and sufficient conditions for explicit $CP2$ conservation for arbitrary $N$. While the conditions themselves can be simply formulated for all $N$, implementing them in practice is not trivial. Making extensive use of Lie algebra and representation theory, we devise an efficient algorithm for detecting whether an arbitrary potential has a $CP2$ symmetry. Thus we are able to check whether a real basis exists, although the possibility of spontaneous $CP$ violation is not addressed in this work.

Throughout this paper we allow $N$ to be arbitrarily large, since our method in principle applies to any number of doublets, although its computational cost increases with $N$. While the 2HDM and 3HDM are currently the most relevant for phenomenology, models with more doublets have received some attention. 4HDMs were studied in e.g.~\cite{Bjorken:1977vt,Arroyo-Urena_2021} and~\cite{Gon_alves_2023}. In the latter article, one doublet couples to quarks and three doublets couple to charged leptons, allowing for flavor changing neutral currents in the leptonic sector, but not in the quark sector. 5HDMs in the context of higher order $CP$s were scrutinized in~\cite{Ivanov_2018}. A 6HDM for Dark Matter was examined in~\cite{Alakhras:2017bbe}, and Grand Unified Theories with eight and nine Higgs doublets were studied in~\cite{Harada:2016fpj} and~\cite{Brahmachari:2008qg}, respectively. Moreover, in the "Private Higgs'' extension of the SM each charged fermion acquire mass from its own Higgs doublet, through $\mathcal{O}(1)$ Yukawa couplings, and is hence another example of a model with $N=9$ Higgs doublets~\cite{Porto_2008,BenTov:2012cx}. The analysis of such models may be facilitated by the general algorithm for $CP2$ detection presented here.

The article is structured as follows: section~\ref{sec:formalism} contains a presentation of the covariant framework for identifying symmetries, which is then applied for deriving a characterization of $CP2$ symmetry, as well as a reminder of Lie algebra theory and proofs of some representation theoretical results for the orthogonal algebra $\So(N)$.  Based on the characterization we derive, algorithms for checking the existence of a $CP2$ symmetry are given in section~\ref{sect:algo}. In section~\ref{sect:examples}, the algorithms are applied to concrete potentials to check for $CP2$. Finally, in section~\ref{sec:Summary} we summarize our results and make final remarks. Additional mathematical results and numerical values for a 7HDM example are found in appendix~\ref{app:extra-results} and~\ref{app:N7-values}.

\section{Formalism}
\label{sec:formalism}
We write the potential for $N$ Higgs $\mathsf{SU}(2)$ doublets $\Phi_i$ in the bilinear formalism~\cite{Maniatis:2015gma}
\begin{equation}
\label{eq:Vbilinear}
V = M_0K_0 + M_aK_a + \Lambda_0 K_0^2 + L_a K_0K_a + \Lambda_{ab}K_aK_b
\end{equation}
where the bilinears $K_\alpha$, $\alpha = 0, \ldots, N^2-1$ are given in terms of the generalized Gell-Mann matrices $\lambda_a$
\begin{equation}\label{E:bilinears}
K_0 = \Phi_i^\dagger\Phi_i \,, \quad K_a = \Phi_i^\dagger (\lambda_a)_{ij} \Phi_j.
\end{equation}
Writing the potential in this manner is advantageous because the bilinears have simple transformation properties under a change of basis
\begin{equation}\label{E:HbasisCh}
\Phi_i \rightarrow U_{ij}\Phi_j,\quad U\in \mathsf{SU}(N),
\end{equation}
with $K_0$ being a singlet while $K_a$ transforms according to the adjoint representation of $\mathsf{SU}(N)$
\begin{equation}
K_0 \rightarrow K_0 \,, \quad K_a\rightarrow R_{ab}(U)K_b
\end{equation}
where
\begin{equation}\label{E:R(U)}
R_{ab}(U) = \frac{1}{2}\text{Tr}(U^\dagger \lambda_a U \lambda_b).
\end{equation}
Since the adjoint representation is the linear action of $\mathsf{SU}(N)$ on the vector space given by its own Lie algebra, all the adjoint vectors which characterize the potential live in $\mathfrak{su}(N)$ which is then the natural setting to derive properties of the potential.

Now, to keep the potential $V$ invariant under the change of basis~\eqref{E:HbasisCh}, the matrix $\Lambda$
has to transform as
\begin{equation}
\label{E:LambdaTrafoU}
\Lambda \to R(U) \Lambda R^T(U).
\end{equation}
The generalized Gell-Mann matrices form a basis for the Lie algebra $\mathfrak{su}(N)$ and satisfy the commutation relations\footnote{In this basis the Killing form is proportional to the identity hence we do not distinguish between upper and lower Lie algebra indices. Moreover, we apply the physicist's definition of a Lie algebra, for mathematicians the mentioned basis would be $\{i \lambda_j \}_{j=1}^{N^2-1}$.}
\begin{equation}
[\lambda_a, \lambda_b] = 2if_{abc}\lambda_c.
\end{equation}
For convenience, we order the generalized Gell-Mann matrices as in~\cite{Solberg:2018aav}, where the antisymmetric matrices appear first. That is
\begin{equation}
\label{eq:custorder}
\lambda_a^T = -\lambda_a \quad \text{for} \quad a=1,\ldots ,k\equiv\frac{N(N-1)}{2}.
\end{equation}
As we will see in section~\ref{sec:framework}, the fact that this subset is equivalent to the defining representation of $\mathfrak{so}(N)$ can be used to derive simple necessary and sufficient conditions for $CP2$ symmetry in NHDMs.

\subsection[Covariant framework for detecting  $CP2$]{Covariant framework for detecting \boldmath $CP2$}
\label{sec:framework}

Let us now describe the setting for characterizing $CP2$ using relations among basis-covariant objects. Our method relies on viewing the adjoint vectors which characterize the potential as elements of $\Su(N)$, thanks to the Lie algebra isomorphism between $\mathfrak{su}(N)$ and $\mathbb{R}^{N^2-1}$ equipped with the F-product from~\cite{deMedeirosVarzielas:2019rrp}
\begin{align}
F : \mathbb{R}^{N^2-1}\times \mathbb{R}^{N^2-1} &\rightarrow \mathbb{R}^{N^2-1}\\
(a,b) &\mapsto f_{ijk}a_i b_j \equiv F_k^{(a,b)}
\end{align}
where $f_{ijk}$ are the structure constants of $\mathfrak{su}(N)$ in the Gell-Mann basis.
The isomorphism is then given by the map
\begin{align}
\label{eq:isomorphism}
\Omega : \mathbb{R}^{N^2-1} &\rightarrow \mathfrak{su}(N)\\
a &\mapsto a_i \lambda_i. 
\end{align}
In what follows, we will denote vectors of $\mathbb{R}^{N^2-1}$ with lower case letters, and the associated $\mathfrak{su}(N)$ matrices by uppercase letters, e.g.~$A \equiv \Omega(a) = a_i\lambda_i$. 
By definition, the generalized Gell-Mann matrices correspond via $\Omega$ to the canonical basis of $\mathbb{R}^{N^2-1}$ i.e.
\begin{equation}
\Omega(e_a) = \lambda_a.
\end{equation}
That $\Omega$ is an isomorphism between the two algebras is easily shown by noticing that
\begin{equation}
\label{eq:Fprod}
F^{(a,b)} = c \iff [A,B] = 2iC. 
\end{equation}
Using this isomorphism we can decompose $\mathbb{R}^{N^2-1}$ into two subspaces
\begin{equation}
\mathbb{R}^{N^2-1} = E_A \oplus E_S,
\end{equation}
with $E_A$ and $E_S$ corresponding respectively to the antisymmetric and symmetric matrices in $\Su(N)$. It is important to note that $E_A$ has a Lie algebra structure since it corresponds to the $\So(N)$ subalgebra while $E_S$ is only a vector space. 

Higgs basis transformations~\eqref{E:HbasisCh} act on $\Su(N)$ as inner automorphisms
\begin{align}
X\to X'=UXU^\dag\quad \text{for} \quad X\in \mathfrak{su}(N)
\end{align}
and hence preserve commutation relations. It follows from~\eqref{eq:Fprod} that F-product relations are also preserved i.e.
\begin{align}\label{E:FprodTransf}
 F^{(a,b)} = c \iff F^{(a',b')} = c',
\end{align}
where $x'=R(U)x$, cf.~\eqref{E:R(U)}. This is simply the statement that F-products relations are vector relations in the adjoint representation of $\mathfrak{su}(N)$. A consequence of~(\ref{E:FprodTransf}) is that if a subspace $V \subset \mathbb{R}^{N^2-1}$ spanned by vectors $\{v_a\}_{a=1}^n$ forms a subalgebra in the sense
\begin{align}
F^{(v_a,v_b)} \in V, \quad \forall\, a,b\in \{1,\ldots ,n\}
\end{align}
then the transformed basis $\{v'_a=R(U)v_a\}_{a=1}^n$ forms the same subalgebra. Since the NHDM potential is completely determined by $\mathbb{R}^{N^2-1}$ vectors, namely $L$, $M$ and the eigenvectors of $\Lambda$, and $\Su(N)$ invariants, any intrinsic property of an arbitrary potential which can be formulated as a set of characteristic vectors spanning a Lie subalgebra, can be verified in any basis. We will now show that $CP2$ symmetry can be characterized in this way.  

\subsubsection[Necessary and sufficient conditions for  $CP2$ symmetry]{Necessary and sufficient conditions for \boldmath $CP2$ symmetry}
\label{sec:$CP2$-condition}

\begin{theorem}
\label{thm:$CP2$}
An NHDM potential admits a $CP2$ symmetry if and only if the following conditions hold
\begin{itemize}
\item $k=\frac{N(N-1)}{2}$ of $\Lambda$'s eigenvectors, $\{v_a\}_{a=1}^k$, form a basis for the defining representation of $\So(N)$
\item $L\cdot v_a = M\cdot v_a = 0, \quad \forall a\in\{1,\ldots ,k\}$.
\end{itemize}
\end{theorem}

\begin{proof}
Suppose a potential~(\ref{eq:Vbilinear}) has a $CP2$ symmetry, then there exists a basis where all the coefficients are real meaning that 
\begin{equation}
\Lambda = \begin{pmatrix}
\label{E:LambdaCP}
B_N & \mathbf{0}\\
\mathbf{0} & A_N 
\end{pmatrix}
\end{equation}
is block diagonal with $B_N$ and $A_N$ arbitrary symmetric matrices of dimension $k \times k$ and $N^2-1-k \times N^2-1-k$, respectively. In this basis, it is evident that $k$ of $\Lambda$'s eigenvectors span $E_A$, and therefore the image of this set by the isomorphism~(\ref{eq:isomorphism}) is a basis for the defining representation of $\So(N)=\text{Span}(\lambda_1,\ldots ,\lambda_k)$. Denote this subset of eigenvectors by $\{t_a\}_{a=1}^k$, it follows that
\begin{equation}
F^{(t_a,t_b)} \in E_A
\end{equation}
i.e.~this subset of eigenvectors closes under the F-product, a property which can be observed in any Higgs basis since F-product relations are basis-independent. In addition, the existence of a real basis implies that the adjoint vectors $L$ and $M$ of the potential are in $E_S$. Thus, the following basis-invariant conditions
\begin{align}
\label{E:cond-LM}
L\cdot t_a = M\cdot t_a = 0, \quad \forall a\in\{1,\ldots ,k\},
\end{align}
must hold. We will sometimes refer to this condition concisely as $LM$-orthogonality.

Conversely, assume the two conditions of the Theorem hold. Indeed, by assumption the representation given by $\{v_a\}_{a=1}^k$ must be equivalent (i.e.~isomorphic) to the defining representation generated by the first $k$ Gell-Mann matrices. Since an equivalence of two hermitian representations with the same underlying vector space is a similarity transformation~\cite{Hall:2000xt} which, as shown in Proposition~\ref{prop:unitary-equiv}, can always be chosen to be unitary, we have  
\begin{align}
\label{E:transf-to-real}
U V_a U^\dagger  \in \text{Span}(\lambda_1,\ldots ,\lambda_k), \quad \forall  a\in\{1,\ldots ,k\}.
\end{align}
The unitary matrix $U$ above is a Higgs basis transformation which brings $\Lambda$ to the block diagonal form~(\ref{E:LambdaCP}), i.e.~it is a transformation to a real basis. To see this note that~(\ref{E:transf-to-real}) when written in terms of adjoint vectors reads
\begin{equation}
R(U)v_a \in E_A, \quad \forall a\in\{1,\ldots ,k\}
\end{equation}
which implies these $k$ eigenvectors span $E_A$ and that the remaining ones $\{R(U)v_a\}_{a=k+1}^{N^2-1}$ span $E_S$. Hence writing $\Lambda$ in terms of its eigenvectors $v_a$ and eigenvalues $\alpha_a$ using its spectral decomposition it follows that
\begin{equation}
R(U)\Lambda R(U)^T = \sum_{i=a}^{N^2-1} \alpha_a R(U) v_a v_a^T R(U)^T
\end{equation}
is block diagonal as in~(\ref{E:LambdaCP}) and hence does not generate complex terms in the potential. Moreover, $LM$-orthogonality in that basis implies that $L$ and $M$ lie in $E_S$ meaning that no complex terms come from these parts of the potential either. Therefore, the two conditions of the Theorem lead to the existence of a real basis.
\end{proof}

Thus the problem of detecting whether a potential has a $CP2$ symmetry is reduced to determining whether $\Lambda$ has $k$ $LM$-orthogonal eigenvectors which form a basis for the defining representation of $\So(N)$.

\subsection{Identifying Lie algebras and representations}
\label{sec:algebra-id}

Applying the characterization of $CP2$ derived in the previous section requires identifying Lie algebras and their representations. Therefore, for the purpose of making the present paper self-contained, we now give a brief reminder of Lie algebra theory focusing on the classification of semisimple Lie algebras and their representations. This presentation is not meant to be exhaustive but simply to introduce characteristics of Lie algebras and how they can be computed in practice. For more complete expositions of Lie algebra and representation theory see e.g.~\cite{Hall:2000xt,fulton1991representation,zee2016group}. 

Given a Lie algebra $\mathfrak{g}$, a Cartan subalgebra $\mathfrak{h}\subset \mathfrak{g}$ is a maximal commuting subalgebra i.e.~a subalgebra of maximal dimension such that 
\begin{equation}
[X,Y]=0, \quad \forall X,Y\in \mathfrak{h}.
\end{equation}
The dimension of $\mathfrak{h}$ is called the rank of $\mathfrak{g}$ and is an important number characterizing a Lie algebra. Denote $\text{dim}(\mathfrak{g})\equiv d$ and $\text{rank}(\mathfrak{g}) \equiv r$ and let $\{H_i\}_{i=1}^r$ be a basis for a Cartan subalgebra. By construction, the adjoint matrices $\ad_{H_i}$ can be simultaneously diagonalized and will have $r$ common nullvectors $h_i$ since
\begin{equation}
\ad_{H_i} h_j = 0 \iff [H_i,H_j] = 0.
\end{equation}
The $d-r$ remaining eigenvectors $e_\alpha$ are called the roots of the algebra and satisfy
\begin{equation}
\ad_{H_i}e_\alpha = \alpha_i e_\alpha \iff [H_i,E_\alpha] = \alpha_i E_\alpha
\end{equation}
while the $d-r$ eigenvalue tuples $\alpha=(\alpha_1,\ldots ,\alpha_r)$, thought of as vectors of $\mathbb{R}^r$, form the so-called root system of $\mathfrak g$. It was shown by Dynkin that semisimple Lie algebras can be classified according to their root system~\cite{Dynkin:1957um}.

Thus an unknown semisimple Lie algebra can be identified if one can compute a Cartan subalgebra for it. This can be done by calculating the nullspace of an adjoint matrix $\ad_{X}$ where $X$ is, by definition, any regular element of $\mathfrak{g}$~\cite{bourbaki2008,de2000lie}. For $\Su(N)$ and its subalgebras,\footnote{If $\mathfrak{g'}\subset\mathfrak{g}$ and $x$ is regular in $\mathfrak{g}$ then $x$ is regular in $\mathfrak{g'}$~\cite{bourbaki2008}.} an element is regular if all its eigenvalues are distinct. Thus a Cartan subalgebra can be computed from a generic element e.g.~randomly sampled.

The dimension of the nullspace of $\ad_X$ then gives the rank of $\mathfrak g$, and the nullvectors $h_i$ provide a basis $H_i$ for a Cartan subalgebra. One can then simultaneously diagonalize the matrices $\ad_{H_i}$ and find the root system. Figure~\ref{fig:B3vsC3} shows how $\So(7)$ and $\mathfrak{sp}(6)$, which have the same dimension and rank, differ by their root system.

\begin{figure}[htbp!]
    \centering
    \begin{subfigure}[t]{0.5\textwidth}
        \centering
        \includegraphics[height=8cm]{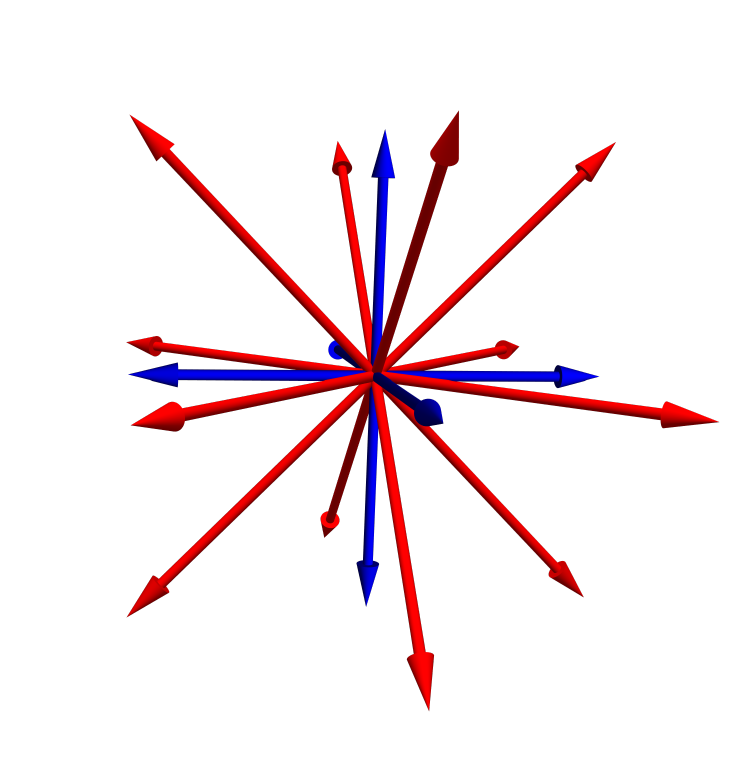}
        \caption{$\So(7)$}
    \end{subfigure}%
    ~ 
    \begin{subfigure}[t]{0.5\textwidth}
        \centering
        \includegraphics[height=8cm]{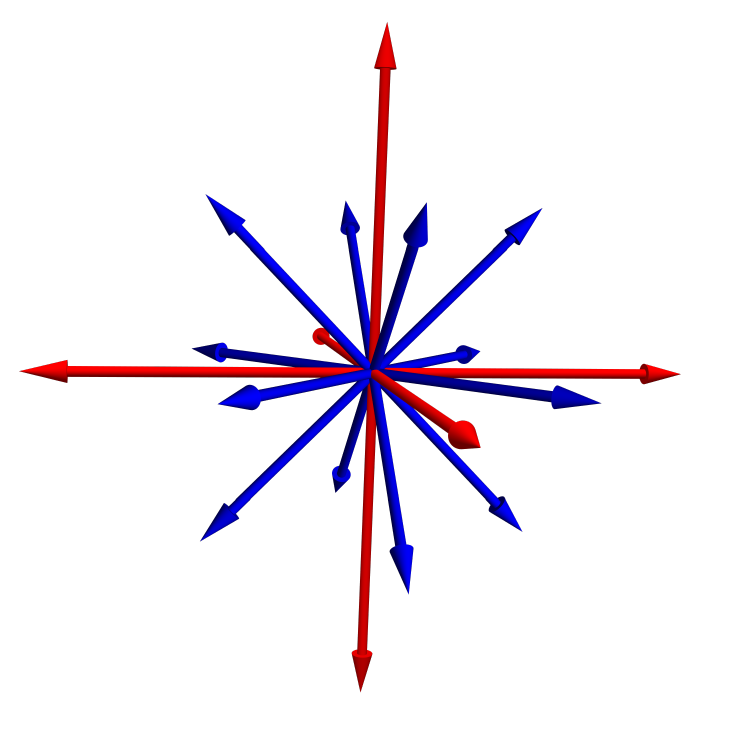}
        \caption{$\mathfrak{sp}(6)$}
    \end{subfigure}
    \caption{Root systems of $\So(7)$ and $\mathfrak{sp}(6)$ with long (short) roots shown in red (blue). The root system distinguishes these two 21-dimensional rank-3 algebras.}
       \label{fig:B3vsC3}
\end{figure}

Once the root system $R$ has been found and an ordered set of positive simple roots $\{\beta_1,\ldots ,\beta_r\} \subset R$ has been chosen~\cite{Hall:2000xt}, the Lie algebra representation at hand can be identified by computing its highest weight $\Upsilon$ which, for an n-dimensional irreducible representation, is a vector of $\mathbb{R}^r$ characterizing that representation~\cite{fulton1991representation}. In the case of a reducible representation there will be one highest weight per irreducible component. To each highest weight corresponds a simultaneous nullvector $v_0$ of the positive simple roots
\begin{equation}
E_{\beta_i}v_0 \equiv 0\, , \quad \forall i\in \{1,\ldots ,r\}.
\end{equation}
The components of the highest weight $\Upsilon=(a_1,\ldots ,a_r)$ in the basis of so-called fundamental weights $\{\omega_i\}_{i=1}^r$~\cite{fulton1991representation}, sometimes called Dynkin labels, are the smallest integers satisfying
\begin{equation}
E_{-\beta_i}^{1+a_i}~ v_0 = 0 \, , \quad i\in \{1,\ldots ,r\}.
\end{equation}

The dimension of an irreducible representation $\Gamma_\Upsilon$ with highest weight $\Upsilon$ is then given by the Weyl dimension
formula~\cite{fulton1991representation, Feger:2019tvk}
\begin{align}\label{E:WeylDimForm}
	\text{dim} (\Gamma_\Upsilon) = \prod_{\alpha \in R^+} 
	\frac{(\alpha,\rho+\Upsilon)}{(\alpha,\rho)}, 
\end{align}
where $R^+\subset R$ is the set of positive roots, $\rho$ is half the sum of the 
positive roots and $(,)$ is the Euclidean inner product. The irreducible representations where one Dynkin label equals one and all other Dynkin labels are zero are called fundamental representations. In particular, the fundamental representation $(1,0,\ldots,0)$ usually corresponds to the defining representation.

\subsubsection[$\So(N)$ subalgebras of  $\Su(N)$]{\boldmath $\So(N)$ subalgebras of \boldmath $\Su(N)$ }

We now prove results about $\So(N)$ subalgebras of $\Su(N)$ that will enable us to devise an algorithm for identifying the defining representation of $\So(N)$ which, as shown in Theorem~\ref{thm:$CP2$}, characterizes $CP2$ symmetry in the NHDM potential.

It can be shown using eq.~\eqref{E:WeylDimForm} that the fundamental representations
of the odd orthogonal Lie algebras $B_r=\So (2r+1)_\mathbb{C}$ with $r\geq 2$, have the following dimensions~\cite{fulton1991representation}:
\begin{align}\label{E:dimOddAlg}
	\text{dim}(\Gamma_{\omega_k})&=\binom{2r+1}{k}, \quad k<r \nonumber\\
	\text{dim}(\Gamma_{\omega_r})&=2^r.
\end{align}
On the other hand, for the even orthogonal algebras $D_r=\So (2r)_\mathbb{C}$ with $r\geq 4$, the fundamental representations have dimensions given by~\cite{fulton1991representation}
\begin{align}\label{E:dimEvenAlg}
	\text{dim}(\Gamma_{\omega_k})&=\binom{2r}{k},  \quad k\leq r- 2 \nonumber\\ 
	\text{dim}(\Gamma_{\omega_{r-1}})&=\text{dim}(\Gamma_{\omega_r})=2^{r-1}.
\end{align}

\begin{prop}
\label{prop:unique-irreps}
If $N\geq 3$ and $N \ne 4, 8$, the defining representation $\mathbf{N}$ is the only irreducible representation of $\mathfrak{so}(N)$ of dimension
$N$, up to equivalence of representations. 
\end{prop}
\begin{proof}
We will here consider the complexification of $\mathfrak{so}(N)$, $\mathfrak{so}(N)_\mathbb{C}$, but there is a bijection $\psi$ between the complex representations of the real and the
complexified Lie algebra, with $\psi(\Pi)(X+iY)=\Pi(X) + i\Pi(Y)$, with $\Pi$ a complex representation of the real Lie algebra, and $X$ and $Y$ elements of the real Lie algebra.
Moreover, $\psi(\Pi)$ is an irreducible representation of the complexified algebra if and only if $\Pi$ is an irreducible representation of the real~algebra~\cite{hall2016lie}.

The non-trivial representations with the smallest dimensions are the fundamental representations 
$\Gamma_{\omega_i}$, where $\omega_i$ is a fundamental weight~\cite{fulton1991representation}.
Recalling that these representations have one Dynkin label being $1$ and the others $0$, this follows from the Weyl dimension formula~\eqref{E:WeylDimForm} and the fact that the dimension of an irreducible representation strictly increases if the any of the Dynkin labels are increased:
\begin{align}\label{E:dimStrictIncr}
	\text{dim}(\Gamma_{(a_1,\ldots,a_i,\ldots ,a_r)})<\text{dim}(\Gamma_{(a_1,\ldots,a_i+1,\ldots ,a_r)}),
\end{align}
where $\Gamma_{(a_1,\ldots ,a_r)}$ is the irreducible representation with Dynkin labels $(a_1,\ldots ,a_r)$, that is, highest weight $\Upsilon=a_1 \omega_1 + \ldots + a_r \omega_r $. Indeed, since the highest weight $\Upsilon$ is always an integral dominant
element, meaning $(\Upsilon,\alpha)\geq 0 $ for each root $\alpha \in R^+$,~\eqref{E:WeylDimForm} yields the inequality in~\eqref{E:dimStrictIncr}. This inequality is strict since the positive roots span all $\mathbb{R}^r$, and hence for all fundamental weights $\omega_i$ there exists a positive root $\alpha$ such that $(\alpha, \omega_i)>0$.

For $\mathfrak{so}(2r+1)_\mathbb{C}$ with $r\geq 2$,  dim$(\Gamma_{\omega_1})=2r+1<$ dim$(\Gamma_{\omega_k})$ for $r>2$ and $1<k\leq r$, cf.~\eqref{E:dimOddAlg}. 
For $r=2$, dim$(\Gamma_{\omega_2})=4$, but the defining representation $\Gamma_{\omega_1}$ is still the unique irreducible representation of
dimension 5, since dim$(\Gamma_{2\omega_2})=10$ and dim$(\Gamma_{(\omega_1+\omega_2)})=16$, according to \texttt{LieART}~\cite{Feger:2019tvk}, where the latter representations correspond to Dynkin labels $(0,2)$ and
$(1,1)$, respectively. Since the dimension increases strictly with increasing Dynkin labels, cf.~\eqref{E:dimStrictIncr}, there are no representations with the same dimension as the defining representation.

In the case $\mathfrak{so}(2r)_\mathbb{C}$ with $r\geq 4$,
the dimension of the defining representation,
dim$(\Gamma_{\omega_1})=2r<\Gamma_{\omega_k}$ for $1<k\leq r-2$ by~\eqref{E:dimEvenAlg}, while for the cases $k>r-2$, dim$(\Gamma_{\omega_1})=2r<
2^{r-1}$ when $r>4$, so the defining representation has the uniquely least dimension among the fundamental representations, for $r>4$.
And again, due to~\eqref{E:dimStrictIncr},
the defining representation of $\mathfrak{so}(2r)_\mathbb{C}$ becomes the unique irreducible representation with dimension $2r$ for $r>4$.
Note that for $r=4$ (i.e.\ $N=8$),~\eqref{E:dimEvenAlg} gives
dim$(\Gamma_{\omega_1})=$dim$(\Gamma_{\omega_3})=$
dim$(\Gamma_{\omega_4})=8$, but this is one of the two cases the Proposition is not valid.

For the remaining values of $N$, i.e.~$N=3,6$, we have Lie algebra isomorphisms between
$\mathfrak{so}(N)_\mathbb{C}$ and Lie algebras with root system $A_r$. The Lie algebra
$\mathfrak{so}(3)_\mathbb{C}\cong \mathfrak{sl}(2)_\mathbb{C}=A_1$  has exactly one irreducible representation of dimension 3
 (this representation does not correspond to a fundamental weight).
Finally, $\mathfrak{so}(6)_\mathbb{C}\cong \mathfrak{sl}(4)_\mathbb{C}=A_3$ has exactly one 6-dimensional representation, $\Gamma_{\omega_2}$. 
 
The defining representation of $\mathfrak{so}(2)$ is not irreducible over $\mathbb{C}$ and is hence not included in this Proposition.
\end{proof}

Both representations and subalgebras of Lie algebras are defined in terms of Lie algebra homomorphisms. The only difference between the two concepts is that a Lie subalgebra always corresponds to an injective (1-1) homomorphism with image in the algebra to which it is a subalgebra, while no such restrictions apply to a representation in general. A subalgebra ${\So(N)}$ of ${\Su(N)}$ is then the same as a faithful representation of ${\So(N)}$ with image in ${\Su(N)}$.
We will by this apply Proposition~\ref{prop:unique-irreps}, which is about representations, to prove a result which describes all possible $\So(N)$ subalgebras of $\Su(N)$, and that is helpful to detect $CP2$ symmetry for any $N$:

\begin{prop}
\label{prop:subalgebras}
The defining representation $\mathbf{N}$ of $\So(N)$ is the only $\So(N)$ subalgebra of $\Su(N)$ up to equivalence (i.e.~conjugation) for $N \geq 3$,
with the following exceptions:   
\begin{itemize}
   \item[] $N=3$: $\mathbf{2}+\mathbf{1}$
     \item[] $N=4$: 
		 $\mathbf{2}+\mathbf{2'} $
     \item[] $N=5$: $\mathbf{4}+\mathbf{1}$
		\item[] $N=6$: 
		$\mathbf{4}+\mathbf{1}+\mathbf{1}$\: \text{and} \: 
		$\mathbf{\bar{4}}+\mathbf{1}+\mathbf{1}$
		\item[]  $N=8$: $\mathbf{8_c}$ \text{and} \:$\mathbf{{8}_s}$
   \end{itemize}
\end{prop}
\begin{proof}
 All subalgebras $\So(N)$ of $\Su(N)$ correspond to a faithful sum of irreducible representations, where the dimensions of the representations sum up to $N$. 
By the discussion of the possible dimensions of irreducible representations
of $\So(N)$ in the proof of Proposition \ref{prop:unique-irreps}, 
the defining representation $\mathbf{N}$ of $\So(N)$ is the only irreducible representation of
	dimension $\leq N$ for $N\geq 9$. 
	
	In the case $N=8$, the dimension formulas~\eqref{E:dimEvenAlg} show there are two additional
	8d irreducible representations. Both of these are faithful, and hence subalgebras, since $\So (8)$ is simple. 
	
	In the case $N=7$ the defining representation is the unique non-trivial representation of lowest dimension, as given by eqs.~\eqref{E:dimOddAlg} and~\eqref{E:dimStrictIncr}. 
	
	For $N=6$, $\So(6) \cong \Su(4)$ has a 4d irreducible representation
$\mathbf{4}$, which also has an inequivalent conjugate representation
$\mathbf{\bar{4}}$. These two inequivalent 4d representations will generate two 6d representations, as displayed in the Proposition. The latter are also inequivalent, because the 
decomposition into a direct sum of irreducible representations is essentially unique (up to a mixing of equivalent summands), since the 
"isotypic" decomposition is unique \cite{procesi2007lie}.
\texttt{LieART}~\cite{Feger:2019tvk} can be applied to check that there are no other
	irreducible representations of $\So(6)$ of dimension less than 6.
 
For $N=5$ there is also a $\mathbf{4}$ since $\So(5)_\mathbb{C} \cong \mathfrak{sp}(4)_\mathbb{C}$, and hence there exists a corresponding 4d irreducible representation 
	of $\So(5)_\mathbb{R}$ due to the 1-1 correspondence between complex representations of real and complex Lie algebras, even though $\So(5)_\mathbb{R} \ncong \mathfrak{sp}(4)_\mathbb{R} \cong \mathfrak{so}(2,3)$. 
	
	Moreover, for $\So(3)\cong\Su(2)$ and $\So(4)\cong \Su(2)\oplus \Su(2)$ we have a 3d representation and a 4d representation, respectively, built from the $\mathbf{2}\cong \mathbf{\bar 2}$ of $\Su(2)$. 
The algebra $\So(4)$ has two inequivalent, irreducible but unfaithful 2d representations, which we will denote $\mathbf{2}$ and $\mathbf{{2}'}$. One of these representations maps the first semisimple component of $\So(4)$ to 
		the Pauli matrices, while the other component is mapped to zero, and vice versa for the other representation.
		Then $\mathbf{2}+\mathbf{2'}$ is a reducible but faithful 4d representation of $\So(4)$, and hence corresponds to a $\So(4)$ subalgebra of $\Su(4)$. Representations
	of $\So(4)$	like $\mathbf{3}+\mathbf{1}$ and $\mathbf{{3'}}+\mathbf{1}$ are not faithful, and do hence not yield $\So(4)$ subalgebras. So  
	$\mathbf{2}+\mathbf{2'}$ is the only possible subalgebra for $N=4$,
	in addition to the defining representation $\mathbf{4}$.
	
	These are the only alternative $N$-dimensional faithful representations for these algebras and hence the only alternative $\So(N)$ subalgebras in $\Su(N)$. Their existence is due to the exceptional isomorphisms among the low-rank simple Lie algebras and in the very special case of $\So(8)$, the high symmetry of the $D_4$ Dynkin diagram~\cite{zee2016group}.

	Finally, all these representations of the compact Lie algebra $\So(N)$ may be written by hermitian matrices~\cite{Cornwell:1985xt}, and will hence exist in $\Su(N)$, just like any representation of a compact Lie group is equivalent to a unitary representation.
\end{proof}
The "exceptional'' subalgebras of Proposition~\ref{prop:subalgebras} are consistent with
the low $N$ subalgebra tables of~\cite{Feger:2019tvk} and~\cite{Lorente:1972xw}.
Ref.~\cite{Lorente:1972xw} lists complex subalgebras of complex, simple algebras, but for
$\Su(N)$, every semisimple complex subalgebra of the complexified algebra $\Su(N)_\mathbb{C}$ will correspond to a semisimple real subalgebra of the compact, real algebra $\Su(N)$, and vice versa.
The latter direction holds for all algebras, while compact $\Su(N)$ have real semisimple subalgebras in 1-1 correspondence with the semisimple complex subalgebras of
  $\Su(N)_\mathbb{C}$. The reason is a complex subalgebra $\mathfrak{h}_\mathbb{C} \subset \Su(N)_\mathbb{C}$ also is a faithful, complex representation of the subalgebra. And since there is a 1-1 correspondence between complex representations of real and complex variants of the algebras, there will also be a corresponding complex representation of the real, compact form $\mathfrak{h}$ of $\mathfrak{h}_\mathbb{C}$. This representation will
		exist in the real algebra $\Su(N)$, since every representation of a compact algebra is equivalent with a Hermitian representation~\cite{Cornwell:1985xt}, i.e.~it is found among the Hermitian matrices of $\Su(N)$.

\section{Algorithms}
\label{sect:algo}
We now present an algorithm which implements the necessary and sufficient condition of Theorem~\ref{thm:$CP2$}, in order to determine if an arbitrary potential has a $CP2$ symmetry. The algorithm works in two steps: first identifying the eigenvectors of $\Lambda$ which are orthogonal to both $L$ and $M$, and then searching for a set of eigenvectors that generates the defining representation of $\So(N)$.

\subsection[Finding all $LM$-orthogonal eigenvectors]{Finding all \boldmath $LM$-orthogonal eigenvectors}
\label{sect:algo-LM}
It is advantageous to start by checking the orthogonality conditions~(\ref{E:cond-LM}) first since that will reduce the number of candidates to be considered when searching for the defining representation of $\So(N)$ among the eigenvectors of $\Lambda$. 
These orthogonality conditions are straightforward to check, but care must be taken when there are eigenvalue subspaces of dimension larger than 1. Indeed, when this is the case, it may be that none of the degenerate eigenvectors are orthogonal to both $L$ and $M$, yet some linear combinations are. Thus the conditions should be checked in an appropriate eigenvector basis where for each eigenvalue subspace $W_\lambda$, corresponding to an eigenvalue $\lambda$, all the independent linear combinations orthogonal to $L$ and $M$ have been extracted. A practical procedure for achieving this is given in Algorithm~\ref{alg:check-LM}.

\begin{algorithm*}[!h]
\caption{Identifying $LM$-orthogonal eigenvectors}
\label{alg:check-LM}

\begin{enumerate}[align=left, label=\fbox{\arabic*}]
\item Given an NHDM potential, compute the $N^2-1$ eigenvectors of $\Lambda$
\item Initialize $B$, a maximal set of $LM$-orthogonal eigenvectors
\item For each eigenvalue subspace $W_\lambda$
\begin{itemize}
\item Solve $\bigg\{\begin{array}{cc}M\cdot X&=0\\L\cdot X&=0\end{array}$ for $X\in W_\lambda$
\item Add an orthonormal basis for the space of solutions to $B$ 
\end{itemize}
\item Return B
\end{enumerate}
\end{algorithm*}

\subsection[Detecting the defining representation of $\So(N)$ in $\Su(N)$]{Detecting the defining representation of  \boldmath $\So(N)$ in \boldmath $\Su(N)$}
\label{sect:algo-detect}

In this section we show an efficient algorithm for determining whether a set of eigenvectors contains a subset which forms a basis for the defining representation of $\So(N)$. 

The strategy is to first determine if a subset of $k$ eigenvectors of $\Lambda$ closes under the F-product i.e.~forms a $k$-dimensional subalgebra. If such a subalgebra exists, we must verify whether it is the $\So(N)$ algebra or some other $k$-dimensional subalgebra, since e.g.~both $\So(5)$ and $\Su(3)\oplus \mathfrak{u}(1)\oplus \mathfrak{u}(1)$ are 10-dimensional subalgebras of $\Su(5)$. For even $N=2r$, computing the rank of the unknown algebra using the method described in~\ref{sec:algebra-id} is enough to unambiguously identify $\So(2r)$, as it is the only\footnote{We have checked this up to $N=22$ (rank 11) by exhausting all the possible semisimple Lie algebras of rank $r$ and dimension $r(2r-1)$ and checking against the $\Su(N)$ subalgebra tables of~\cite{Feger:2019tvk}. For $N\geq 24$, there may be subalgebras of same dimension and rank as $\So(N)$ and one has to look at the root systems.} subalgebra of $\Su(2r)$ with dimension $k$ and rank $r$. For odd $N=2r+1$, $\Su(2r+1)$ always has, in addition to $\So(2r+1)$, at least an $\mathfrak{sp}(2r)$ subalgebra which has the same dimension and rank. Thus, beyond $r=1$ and $r=2$ where one has the isomorphisms $\So(3)\cong \mathfrak{sp}(2)$ and $\So(5) \cong \mathfrak{sp}(4)$, the root system of the unknown algebra must be computed in order to establish that it is $\So(2r+1)$. 

Having identified an $\So(N)$ subalgebra, it remains to check whether it corresponds to the defining representation. As shown in Proposition~\ref{prop:subalgebras}, unless $N=3,4,5,6,8$ the defining representation is the only $\So(N)$ subalgebra in $\Su(N)$ and we are done. For the remaining special values, the representation must be identified by computing its highest weights.

The complete procedure for arbitrary $N$ is given in Algorithm~\ref{alg:check-defining}.

\begin{algorithm*}[!h]
\caption{Checking if eigenvectors generate the defining representation of $\So(N)$}
\label{alg:check-defining}

\begin{enumerate}[align=left, label=\fbox{\arabic*}]
\item Input a set of orthonormal eigenvectors.
\item If there exists a subset of $k$ eigenvectors $\{v_a\}$ forming a basis for a subalgebra $\mathfrak{g}~\subset~\Su(N)$ (see \ref{sec:subalgebra-test}), proceed. Else, return False. 
\item Compute $r\equiv \text{Rank}(\mathfrak{g})$ (see~\ref{sec:algebra-id}).
\begin{itemize}
\item If $N$ is odd and $r=\frac{N-1}{2}$, proceed.
\item If $N$ is even and $r=\frac{N}{2}$, go to \fbox{5} if $r\leq 11$, else proceed.
\item Else return False.
\end{itemize}
\item If the root system of $\mathfrak{g}$ is that of $\So(N)$ (see \ref{sec:algebra-id}), proceed. Else, return False.
\item If $N\neq 3,4,5,6,8$, return True. Else, compute the highest weights of the $N$-dimensional representation of $\So(N)$.
\item If the highest weight is that of the defining representation (see \ref{sec:algebra-id}), return True. Else return False.
\end{enumerate}
\end{algorithm*}

\subsubsection{Testing for a subalgebra}
\label{sec:subalgebra-test}

In order to implement Algorithm~\ref{alg:check-defining}, one must be able to detect whether a subset of vectors forms a basis for a subalgebra, which can be done by considering the structure constants. Let $\{v_i\}_{i=1}^{N^2-1}$ be an orthonormal set of eigenvectors of the real symmetric matrix $\Lambda$, the structure constants in that basis of $\Su(N)$ are given by
\begin{equation}
\label{eq:struct-eigen}
Z_{ijk} \equiv F^{(v_i,v_j)} \cdot v_k = \frac{-i}{4}\Tr\big([V_i,V_j]V_k\big),
\end{equation}
and the closure of a subalgebra generated by a subset $\{v_a\}_{a\in I}$, $I\subset \{1,\ldots ,N^2-1\}$ means that
\begin{equation}
\label{eq:subalg-pattern}
Z_{abc} = 0 \quad \forall a,b \in I,\, c\notin I.
\end{equation}
The presence of such a pattern in the structure constants is typically easy to detect, except in the isolated cases where the structure constants array $Z$ is sparse. This happens for instance when the matrix $\Lambda$ is exactly diagonal in some basis in which case we have $Z_{ijk}=f_{ijk}$ and it becomes difficult to identify the pattern~(\ref{eq:subalg-pattern}) among the many zeroes of $Z$, without resorting to brute-force checking all the possible 
subsets of eigenvectors. In the context of a uniform numerical scan this is not an issue, since parameter points corresponding to exactly diagonal $\Lambda$ matrices are a measure zero parameter space subset, and hence  in practice are almost never sampled. In a more general setting, one can deduce from the sparsity of $Z$ that the potential takes a very simple form in some basis, and thus is likely to have large symmetries. A case-by-case analysis may be necessary to identify these symmetries when the number of doublets is too large to check for the pattern~(\ref{eq:subalg-pattern}) using brute-force. 

\subsection{Numerical considerations}

Some comments about the practical implementation of $CP2$ detection by means of Algorithms~\ref{alg:check-LM} and~\ref{alg:check-defining} are in order. First, all the steps in these algorithms are linear algebra computations which can, in principle, be carried out analytically. However, a complete analytic treatment would require very simple expressions for the eigenvectors of $\Lambda$ which is unlikely to be the case for non-trivial potentials. Thus, in practice, a numerical implementation of the algorithms is most relevant.   

Secondly, it is often the case (e.g.~in a uniform parameter scan) that a symmetry cannot truly be exact. This can be due to the symmetric subset of the parameter space having measure zero, or simply finite numerical precision. Either way, an appropriate numerical tolerance has to be defined such that parameter points which are sufficiently close to exact symmetry are considered symmetric. We want to emphasize that Algorithms~\ref{alg:check-LM} and~\ref{alg:check-defining} can be implemented with such a tolerance in order to detect parameter points close to exact $CP2$ symmetry. Indeed, if the tolerance is encoded by a small number $\epsilon$, then it suffices to neglect all numbers smaller than $\epsilon$ in the numerical computations. 

Lastly, even though large values of $N$ have limited practical applications to e.g.~phenomenology, one might wonder about the computational cost of Algorithm~\ref{alg:check-defining} and how it scales with $N$. The most expensive step is checking for a subalgebra since that requires the computation of the structure constants~(\ref{eq:struct-eigen})\footnote{With the exception of the pathological cases discussed in~\ref{sec:subalgebra-test}.} for which the required number of operations scales as $N^{12}$. While the computation time increases fast, the presence or absence of $CP2$ can be established almost instantaneously for $N=3$ and $N=4$ doublets which are, arguably, the most important use cases.

\section{Examples}
\label{sect:examples}
We now illustrate how our $CP2$ detection method, summarized in Algorithm~\ref{alg:$CP2$}, can be applied concretely to determine whether a particular instance of a NHDM potential has a $CP2$ symmetry.

\begin{algorithm*}[!h]
\caption{Detecting $CP2$ symmetry}
\label{alg:$CP2$}

\begin{enumerate}[align=left, label=\fbox{\arabic*}]
\item Let $B$ be the result of Algorithm~\ref{alg:check-LM}. If $B$ contains less than $k$ eigenvectors, return False. Otherwise proceed.
\item If Algorithm~\ref{alg:check-defining} applied to $B$ returns True then return True. Else return False.
\end{enumerate}
\end{algorithm*}

\subsection[$N=3:$ The Ivanov-Silva potential]{\boldmath $N=3:$ The Ivanov-Silva potential}
\label{sec:N3}
The Ivanov-Silva potential is an example of a model with a $CP4$ symmetry but no $CP2$ symmetry, and hence a $CP$-conserving potential without a real basis~\cite{Ivanov_2016}. Consider a particular numerical instance of this potential in a basis where the existence of neither a $CP2$ or $CP4$ symmetry is obvious, given by the following parameters

\begin{align}
\label{eq:N3-Lambda}
\Lambda &= \left(
\begin{array}{cccccccc}
 -22 & -4 \sqrt{3} & 2 & 0 & 2 \sqrt{3} & -12 & 2 \sqrt{3} & 8 \\
 -4 \sqrt{3} & -14 & -6 \sqrt{3} & 4 \sqrt{3} & 6 & 2 \sqrt{3} & 4 & -8 \sqrt{3} \\
 2 & -6 \sqrt{3} & -2 & 0 & 6 \sqrt{3} & -18 & 2 \sqrt{3} & 0 \\
 0 & 4 \sqrt{3} & 0 & 16 & 0 & -4 & 0 & 0 \\
 2 \sqrt{3} & 6 & 6 \sqrt{3} & 0 & 10 & 6 \sqrt{3} & 6 & 0 \\
 -12 & 2 \sqrt{3} & -18 & -4 & 6 \sqrt{3} & -10 & 4 \sqrt{3} & -24 \\
 2 \sqrt{3} & 4 & 2 \sqrt{3} & 0 & 6 & 4 \sqrt{3} & -18 & -\frac{8}{\sqrt{3}} \\
 8 & -8 \sqrt{3} & 0 & 0 & 0 & -24 & -\frac{8}{\sqrt{3}} & -\frac{40}{3} \\
\end{array}
\right)\\
L &= \left(
\begin{array}{cccccccc}
 -\frac{16}{\sqrt{3}} &
 0 &
 0 &
 0 &
 0 &
 0 &
 \frac{16}{3} &
 \frac{32}{3 \sqrt{3}} 
\end{array}
\right)\\
M &= \left(
\begin{array}{cccccccc}
 -8 \sqrt{3} &
 0 &
 0 &
 0 &
 0 &
 0 &
 8 &
 \frac{16}{\sqrt{3}} 
\end{array}
\right)\\
\label{eq:N3-M0}
\Lambda_0 &= -\frac{64}{9}\, ,\quad M_0 = -\frac{112}{3}.
\end{align}

Applying Algorithm~\ref{alg:$CP2$}, we start by looking for a maximal set of $LM$-orthogonal eigenvectors which may contain up to seven elements since in this particular case $L$ and $M$ happen to be colinear. Now, $\Lambda$ has two 2d eigenvalue subspaces, both being $LM$-orthogonal, while all remaining 1d eigenvalue spaces except one, are $LM$-orthogonal. Thus one finds that the eigenvectors associated with the eigenvalues       
\begin{equation}
-48,\,-8\sqrt{5},\,-8,\,8\sqrt{5},\,32,
\end{equation}
where eigenvalues $\pm 8\sqrt{5}$ have multiplicity 2, form a set of maximal $LM$-orthogonal eigenvectors which we denote $\{v_a\}_{a=1}^7$. Eq.~(\ref{eq:N3-Z}) below shows the structure constants $Z_{(ab)c}$ in this basis of eigenvectors, arranged as a matrix with non-zero elements denoted by $\ast$.
\begin{equation}
\label{eq:N3-Z}
Z_{(ab)c}=F^{(v_a,v_b)}\cdot v_c = 
\begin{blockarray}{cccccccc}
\textcolor{blue}{\scriptstyle{v_1}} & \scriptstyle{v_2} & \scriptstyle{v_3} & \textcolor{blue}{\scriptstyle{v_4}} & \scriptstyle{v_5} & \scriptstyle{v_6} & \textcolor{blue}{\scriptstyle{v_7}} \\
\begin{block}{(ccccccc)c}
 0 & 0 & \ast & 0 & \ast & \ast & 0 & \scriptstyle{F^{(v_1,v_2)}} \\
 0 & \ast & 0 & 0 & \ast & \ast & 0 & \scriptstyle{F^{(v_1,v_3)}} \\
 0 & 0 & 0 & 0 & 0 & 0 & \ast & \textcolor{blue}{\scriptstyle{F^{(v_1,v_4)}}} \\
 0 & \ast & \ast & 0 & 0 & \ast & 0 & \scriptstyle{F^{(v_1,v_5)}} \\
 0 & \ast & \ast & 0 & \ast & 0 & 0 & \scriptstyle{F^{(v_1,v_6)}} \\
 0 & 0 & 0 & \ast & 0 & 0 & 0 & \textcolor{blue}{\scriptstyle{F^{(v_1,v_7)}}}\\
 \ast & 0 & 0 & \ast & 0 & 0 & \ast & \scriptstyle{F^{(v_2,v_3)}} \\
 0 & 0 & \ast & 0 & \ast & 0 & 0 & \scriptstyle{F^{(v_2,v_4)}} \\
 \ast & 0 & 0 & \ast & 0 & 0 & \ast & \scriptstyle{F^{(v_2,v_5)}} \\
 \ast & 0 & 0 & 0 & 0 & 0 & \ast & \scriptstyle{F^{(v_2,v_6)}} \\
 0 & 0 & \ast & 0 & \ast & \ast & 0 & \scriptstyle{F^{(v_2,v_7)}} \\
 0 & \ast & 0 & 0 & 0 & \ast & 0 & \scriptstyle{F^{(v_3,v_4)}} \\
 \ast & 0 & 0 & 0 & 0 & 0 & \ast & \scriptstyle{F^{(v_3,v_5)}} \\
 \ast & 0 & 0 & \ast & 0 & 0 & \ast & \scriptstyle{F^{(v_3,v_6)}} \\
 0 & \ast & 0 & 0 & \ast & \ast & 0 & \scriptstyle{F^{(v_3,v_7)}} \\
 0 & \ast & 0 & 0 & 0 & \ast & 0 & \scriptstyle{F^{(v_4,v_5)}} \\
 0 & 0 & \ast & 0 & \ast & 0 & 0 & \scriptstyle{F^{(v_4,v_6)}} \\
 \ast & 0 & 0 & 0 & 0 & 0 & 0 & \textcolor{blue}{\scriptstyle{F^{(v_4,v_7)}}} \\
 \ast & 0 & 0 & \ast & 0 & 0 & \ast & \scriptstyle{F^{(v_5,v_6)}} \\
 0 & \ast & \ast & 0 & 0 & \ast & 0 & \scriptstyle{F^{(v_5,v_7)}} \\
 0 & \ast & \ast & 0 & \ast & 0 & 0 & \scriptstyle{F^{(v_6,v_7)}} \\
\end{block}
\end{blockarray}
\end{equation}
It is now easy to isolate which subsets of 3 eigenvectors may close under the F-product. Indeed two eigenvectors can only be a basis for a 3d subalgebra if their F-product has components along no more than one other eigenvector. By implementing this criteria one avoids to blindly check all ${\small \binom{7}{3}}=35$ possible subsets for closure under the F-product. In the example at hand, this analysis reveals that 
\begin{equation}
\label{eq:N3-subalg-ex}
(V_1,\,V_4,\,V_7)
\end{equation}
forms a 3d, rank-1 subalgebra of $\Su(3)$ which must then be $\So(3)$. It remains to identify which representation is generated by~(\ref{eq:N3-subalg-ex}) by computing the Dynkin labels. Without loss of generality, take $V_7$ as the basis for a Cartan subalgebra and let $\tilde V_i$ be the matrices in the basis where $V_7$ is diagonal. The only positive simple root is then given by 
\begin{equation}
E_+ = \tilde V_1 + i \tilde V_4
\end{equation}
and has two orthogonal nullvectors $u_1$ and $u_2$ satisfying
\begin{align}
E_- u_1 &= 0 \\
E_-^2 u_2 &= 0
\end{align}
showing that the representation has two highest weights with Dynkin labels 0 and 1. The potential given by~(\ref{eq:N3-Lambda})--(\ref{eq:N3-M0})   has therefore, as expected, no $CP2$ symmetry since $\Lambda$ has three eigenvectors forming a basis for the representation $\mathbf{2+1}$ of $\So(3)$ while it is $\mathbf{3}$ which corresponds to $CP2$. Actually the representation $\mathbf{2+1}$, accompanied by the $LM$-orthogonality conditions, corresponds to a different block structure which partially characterizes $CP4$ in 3HDMs~\cite{Ivanov:2018ime}.

\subsection[$N=4: \:\mathbb{Z}_6$-symmetric potential]{\boldmath $N=4: \:\mathbb{Z}_6$-symmetric potential}
\label{sec:N4}

As an example with four doublets, we now study a $\mathbb{Z}_6$-symmetric 4HDM potential which has both non-$CP2$ and $CP2$-symmetric parameter points~\cite{Shao:2023oxt}. Consider the following numerical instance 
\begin{align}
\label{eq:N4-Lambda}
\Lambda &= 
\left(
\begin{array}{ccccccccccccccc}
 \scriptstyle{\frac{3}{8}} & \scriptstyle{0} & \scriptstyle{0} & \scriptstyle{0} & \scriptstyle{0} & \scriptstyle{-\frac{7}{8}} & \scriptstyle{0} & \scriptstyle{0} & \scriptstyle{0} & \scriptstyle{0} & \scriptstyle{0} & \scriptstyle{0} & \scriptstyle{0} & \scriptstyle{0} & \scriptstyle{0} \\
 \scriptstyle{0} & \scriptstyle{\frac{1}{4}} & \scriptstyle{-\frac{\sqrt{3}}{16}} & \scriptstyle{-\frac{\sqrt{3}}{16}} & \scriptstyle{\frac{1}{4}} & \scriptstyle{0} & \scriptstyle{\frac{\sqrt{3}}{16}} & \scriptstyle{0} & \scriptstyle{0} & \scriptstyle{0} & \scriptstyle{0} & \scriptstyle{\frac{\sqrt{3}}{16}} & \scriptstyle{\frac{1}{16}} & \scriptstyle{\frac{\sqrt{3}}{16}} & \scriptstyle{-\frac{\sqrt{\frac{3}{2}}}{8}} \\
 \scriptstyle{0} & \scriptstyle{-\frac{\sqrt{3}}{16}} & \scriptstyle{\frac{1}{4}} & \scriptstyle{-\frac{1}{2}} & \scriptstyle{\frac{\sqrt{3}}{16}} & \scriptstyle{0} & \scriptstyle{-\frac{1}{16}} & \scriptstyle{0} & \scriptstyle{0} & \scriptstyle{0} & \scriptstyle{0} & \scriptstyle{\frac{3}{16}} & \scriptstyle{\frac{\sqrt{3}}{16}} & \scriptstyle{-\frac{1}{16}} & \scriptstyle{\frac{1}{8 \sqrt{2}}} \\
 \scriptstyle{0} & \scriptstyle{-\frac{\sqrt{3}}{16}} & \scriptstyle{-\frac{1}{2}} & \scriptstyle{\frac{1}{4}} & \scriptstyle{\frac{\sqrt{3}}{16}} & \scriptstyle{0} & \scriptstyle{-\frac{1}{16}} & \scriptstyle{0} & \scriptstyle{0} & \scriptstyle{0} & \scriptstyle{0} & \scriptstyle{\frac{3}{16}} & \scriptstyle{\frac{\sqrt{3}}{16}} & \scriptstyle{-\frac{1}{16}} & \scriptstyle{\frac{1}{8 \sqrt{2}}} \\
 \scriptstyle{0} & \scriptstyle{\frac{1}{4}} & \scriptstyle{\frac{\sqrt{3}}{16}} & \scriptstyle{\frac{\sqrt{3}}{16}} & \scriptstyle{\frac{1}{4}} & \scriptstyle{0} & \scriptstyle{-\frac{\sqrt{3}}{16}} & \scriptstyle{0} & \scriptstyle{0} & \scriptstyle{0} & \scriptstyle{0} & \scriptstyle{-\frac{\sqrt{3}}{16}} & \scriptstyle{-\frac{1}{16}} & \scriptstyle{-\frac{\sqrt{3}}{16}} & \scriptstyle{\frac{\sqrt{\frac{3}{2}}}{8}} \\
 \scriptstyle{-\frac{7}{8}} & \scriptstyle{0} & \scriptstyle{0} & \scriptstyle{0} & \scriptstyle{0} & \scriptstyle{\frac{3}{8}} & \scriptstyle{0} & \scriptstyle{0} & \scriptstyle{0} & \scriptstyle{0} & \scriptstyle{0} & \scriptstyle{0} & \scriptstyle{0} & \scriptstyle{0} & \scriptstyle{0} \\
 \scriptstyle{0} & \scriptstyle{\frac{\sqrt{3}}{16}} & \scriptstyle{-\frac{1}{16}} & \scriptstyle{-\frac{1}{16}} & \scriptstyle{-\frac{\sqrt{3}}{16}} & \scriptstyle{0} & \scriptstyle{-\frac{1}{8}} & \scriptstyle{0} & \scriptstyle{0} & \scriptstyle{0} & \scriptstyle{0} & \scriptstyle{-\frac{3}{8}} & \scriptstyle{-\frac{\sqrt{3}}{16}} & \scriptstyle{\frac{1}{16}} & \scriptstyle{-\frac{1}{8 \sqrt{2}}} \\
 \scriptstyle{0} & \scriptstyle{0} & \scriptstyle{0} & \scriptstyle{0} & \scriptstyle{0} & \scriptstyle{0} & \scriptstyle{0} & \scriptstyle{\frac{11}{8}} & \scriptstyle{0} & \scriptstyle{0} & \scriptstyle{-\frac{1}{8}} & \scriptstyle{0} & \scriptstyle{0} & \scriptstyle{0} & \scriptstyle{0} \\
 \scriptstyle{0} & \scriptstyle{0} & \scriptstyle{0} & \scriptstyle{0} & \scriptstyle{0} & \scriptstyle{0} & \scriptstyle{0} & \scriptstyle{0} & \scriptstyle{\frac{1}{2}} & \scriptstyle{1} & \scriptstyle{0} & \scriptstyle{0} & \scriptstyle{0} & \scriptstyle{0} & \scriptstyle{0} \\
 \scriptstyle{0} & \scriptstyle{0} & \scriptstyle{0} & \scriptstyle{0} & \scriptstyle{0} & \scriptstyle{0} & \scriptstyle{0} & \scriptstyle{0} & \scriptstyle{1} & \scriptstyle{\frac{1}{2}} & \scriptstyle{0} & \scriptstyle{0} & \scriptstyle{0} & \scriptstyle{0} & \scriptstyle{0} \\
 \scriptstyle{0} & \scriptstyle{0} & \scriptstyle{0} & \scriptstyle{0} & \scriptstyle{0} & \scriptstyle{0} & \scriptstyle{0} & \scriptstyle{-\frac{1}{8}} & \scriptstyle{0} & \scriptstyle{0} & \scriptstyle{\frac{11}{8}} & \scriptstyle{0} & \scriptstyle{0} & \scriptstyle{0} & \scriptstyle{0} \\
 \scriptstyle{0} & \scriptstyle{\frac{\sqrt{3}}{16}} & \scriptstyle{\frac{3}{16}} & \scriptstyle{\frac{3}{16}} & \scriptstyle{-\frac{\sqrt{3}}{16}} & \scriptstyle{0} & \scriptstyle{-\frac{3}{8}} & \scriptstyle{0} & \scriptstyle{0} & \scriptstyle{0} & \scriptstyle{0} & \scriptstyle{\frac{3}{8}} & \scriptstyle{-\frac{\sqrt{3}}{16}} & \scriptstyle{\frac{1}{16}} & \scriptstyle{-\frac{1}{8 \sqrt{2}}} \\
 \scriptstyle{0} & \scriptstyle{\frac{1}{16}} & \scriptstyle{\frac{\sqrt{3}}{16}} & \scriptstyle{\frac{\sqrt{3}}{16}} & \scriptstyle{-\frac{1}{16}} & \scriptstyle{0} & \scriptstyle{-\frac{\sqrt{3}}{16}} & \scriptstyle{0} & \scriptstyle{0} & \scriptstyle{0} & \scriptstyle{0} & \scriptstyle{-\frac{\sqrt{3}}{16}} & \scriptstyle{\frac{1}{2}} & \scriptstyle{\frac{1}{4 \sqrt{3}}} & \scriptstyle{-\frac{1}{2 \sqrt{6}}} \\
 \scriptstyle{0} & \scriptstyle{\frac{\sqrt{3}}{16}} & \scriptstyle{-\frac{1}{16}} & \scriptstyle{-\frac{1}{16}} & \scriptstyle{-\frac{\sqrt{3}}{16}} & \scriptstyle{0} & \scriptstyle{\frac{1}{16}} & \scriptstyle{0} & \scriptstyle{0} & \scriptstyle{0} & \scriptstyle{0} & \scriptstyle{\frac{1}{16}} & \scriptstyle{\frac{1}{4 \sqrt{3}}} & \scriptstyle{\frac{1}{2}} & \scriptstyle{\frac{1}{2 \sqrt{2}}} \\
 \scriptstyle{0} & \scriptstyle{-\frac{\sqrt{\frac{3}{2}}}{8}} & \scriptstyle{\frac{1}{8 \sqrt{2}}} & \scriptstyle{\frac{1}{8 \sqrt{2}}} &
   \scriptstyle{\frac{\sqrt{\frac{3}{2}}}{8}} & \scriptstyle{0} & \scriptstyle{-\frac{1}{8 \sqrt{2}}} & \scriptstyle{0} & \scriptstyle{0} & \scriptstyle{0} & \scriptstyle{0} & \scriptstyle{-\frac{1}{8 \sqrt{2}}} &
   \scriptstyle{-\frac{1}{2 \sqrt{6}}} & \scriptstyle{\frac{1}{2 \sqrt{2}}} & \scriptstyle{\frac{1}{4}} \\
\end{array}
\right) \\
L &= M = 0 \\
\Lambda_0 &= \frac{5}{4} \, ,\quad M_0 = -1
\end{align}
which, in the notation of~\cite{Shao:2023oxt}, corresponds to the couplings taking on the values 
\begin{align}
m^2 &=1, \quad \Lambda = 1, \quad \Lambda' = 2, \quad \Lambda'' = 3, \quad \tilde\Lambda' = 4, \quad \tilde\Lambda'' = -\frac{1}{2}, \quad \lambda_1 = i,\notag\\
\lambda_2 &=i, \quad \lambda_3 =e^{i\frac{2\pi}{3}}, \quad \lambda_4 =1, \quad   \lambda_5 = 2, \quad
\end{align}
and transformed to a different basis. 

Since $L=M=0$, the orthogonality conditions are automatically satisfied by all eigenvectors of~(\ref{eq:N4-Lambda}) and we proceed to compute the $\Su(N)$ structure constants $Z_{(ij)k}$ in the basis of eigenvectors, in order to look for 6d subalgebras. Due to its large size, displaying the $Z$-matrix is impractical and not very illuminating, therefore we omit it. Nevertheless, the subalgebra search is easily done using a computer program to implement the strategies explained in section~\ref{sec:N3}, and it is found that two sets of eigenvectors generate 6d subalgebras of rank 2, for which the only
candidate is $\So(4)$. Thus it is not necessary to compute the root system, and it only remains to identify which 4d representations have been found. From Proposition~\ref{prop:subalgebras}, the only possibilities are $\mathbf{4}$, i.e.~the defining representation, and the reducible representation $\mathbf{2+2'}$. Computing the Dynkin labels of both representations as described in section~\ref{sec:algebra-id} one finds
\begin{equation}
(1,0)+(0,1)\sim \mathbf{2+2'},
\end{equation}
which is the aforementioned reducible representation, and
\begin{equation}\label{E:114}
(1,1)\sim \mathbf{4}
\end{equation}
which is the defining representation, as can be verified using e.g.~\texttt{LieART}~\cite{Feger:2019tvk}. The detection of the defining representation~\eqref{E:114} implies the existence of a $CP2$ symmetry for this parameter point.
\subsection[$N=7$]{\boldmath $N=7$}
\label{sec:N7}
As a last example which, while mostly academic, shows the power of this method for $CP2$ detection, we apply Algorithm~\ref{alg:$CP2$} to a 7HDM potential whose parameter values are given in appendix~\ref{app:N7-values}. From a Lie algebraic perspective, $N=7$ is interesting as it is the first value where there exists a semisimple Lie algebra with the same dimension and rank as $\So(N)$, but which is not isomorphic to it, namely $\mathfrak{sp}(6)$.

Algorithm~\ref{alg:check-LM} reveals that a maximal set of orthonormal eigenvectors satisfying $LM$-orthogonality has 41 elements. With so many candidate eigenvectors, searching for the subalgebra pattern~(\ref{eq:subalg-pattern}) in $\Su(7)$ starts to become computationally expensive. For reference, running our implementation of Algorithm~\ref{alg:$CP2$} on an ordinary computer it takes less than a minute to find that 21 eigenvectors close under the F-product, forming a 21d rank-3 subalgebra. The root systems of the two possible algebras, $\So(7)$ and $\mathfrak{sp}(6)$, which differ only by the lengths of the roots, are shown in figure~\ref{fig:B3vsC3}. In the example at hand one finds that the root system of the algebra to be identified is in fact that of $\mathfrak{sp}(6)$, and hence the corresponding potential has no $CP2$ symmetry.

\section{Summary}
\label{sec:Summary}

We have derived necessary and sufficient conditions for an NHDM potential to admit a $CP2$ symmetry, which are formulated as relations among vectors that transform according to the adjoint representation under an $\SU(N)$ change of doublet basis. Such vectors can naturally be thought of elements of $\Su(N)$, which allows one to use the Lie algebra structure to verify basis-invariant properties such as being related to a particular subspace in the adjoint space. In the case of $CP2$, the relevant subspace actually corresponds to a $\So(N)$ subalgebra of $\Su(N)$ and the main task for detecting this symmetry is checking whether a subset of eigenvectors of the bilinear quadratic form $\Lambda$ generates the defining representation of $\So(N)$. By considering all the $k$-dimensional subalgebras of $\Su(N)$ and all the $N$-dimensional representations of $\So(N)$ we developed an optimized computable algorithm for this task. The complete algorithm for detecting $CP2$ works in principle for any number of doublets $N$, and is only limited by computational cost. We find that, running our algorithm on a regular desktop computer, a generic parameter space point can be labelled $CP2$-conserving or $CP2$-violating in less than a minute for $N\leq 7$. However, when a $CP2$ symmetry exists, finding a real basis explicitly in general remains out of reach.

\section*{Acknowledgements}

RP is grateful to Igor P.~Ivanov, Celso C.~Nishi and Andreas Trautner for stimulating
discussions and helpful comments which enhanced his understanding of
covariants-based methods.
MS would like to thank Igor P.~Ivanov for suggesting ref.~\cite{deMedeirosVarzielas:2019rrp} and the covariant method regarding detection of the custodial symmetry, which also had a great influence on the present work.

\appendix
\section{Additional mathematical results}
\label{app:extra-results}

\begin{prop}
\label{prop:unitary-equiv}
Let $\{X_a\}$ and $\{Y_a\}$ be bases of two equivalent, Hermitian, irreducible and complex representations of the same Lie algebra i.e.~there exists an invertible $S$ such that $Y_a = SX_aS^{-1}$ for all $a$. Then $S$ can be chosen to be special unitary.
\end{prop}
\begin{proof}
We have
\begin{align*}
&Y_a^\dagger = (S^{-1})^\dagger X_a^\dagger S^\dagger =  (S^{-1})^\dagger X_a S^\dagger = Y_a = SX_aS^{-1} \\
&\implies X_a S^\dagger S = S^\dagger S X_a
\end{align*}
for all $a$. The matrix $S^\dagger S$ thus commutes with all the elements of an irreducible complex representation and hence, by Schur's lemma, $S^\dagger S$ must be proportional to the identity. Let $\lambda$ be the proportionality constant, then $\lambda > 0$ since it is an eigenvalue of the positive definite matrix $S^\dagger S$. Then the rescaled matrix ${S}/{\sqrt{\lambda}}$ is unitary, and may always be written as a special unitary matrix $U$ times a complex phase
$e^{i\theta}$. Hence $S=\sqrt{\lambda}\, e^{i\theta} U$ while $S^{-1}= U^\dagger /(\sqrt{\lambda}\, e^{i\theta})$, and the result follows.

\end{proof}
An immediate consequence of Proposition~\ref{prop:unitary-equiv} is then the following:
\begin{prop}
\label{prop:unitary-equivRedIrred}
Two equivalent, irreducible representations of $\mathfrak{so}(N)$ contained in $\mathfrak{su}(N)$
may always be related by a similarity transformation given by a unitary matrix $U$.
\end{prop}

\section{Parameter values for the  $N=7$ numerical example}
\label{app:N7-values}
Below are the numerical values for the parameter point used in the example analyzed in section~\ref{sec:N7}. All the non-zero elements are listed, except those which can be obtained by symmetry of $\Lambda$.
{\allowdisplaybreaks\setlength{\lineskip}{3pt}
\setlength{\lineskiplimit}{4pt}\begin{alignat*}{8}
&\Lambda_{1,3} &&= -1, \quad &&\Lambda_{4,6} &&= \frac{1}{2}, \quad &&\Lambda_{7,19} &&\!= \frac{5}{2}, \quad &&\Lambda_{7,43}&&\!= \frac{1}{2\sqrt{2}}, \\
&\Lambda_{7,44} &&= -\frac{1}{2 \sqrt{6}}, \quad &&\Lambda_{7,45} &&= -\frac{1}{4 \sqrt{3}}, \quad &&\Lambda_{7,46} &&\!= -\frac{\sqrt{5}}{4}, \\
&\Lambda_{8,20} &&= 1, \quad &&\Lambda_{10,13} &&= \frac{3}{2}, \quad &&\Lambda_{11,15} &&= -\frac{1}{4}, \quad &&\Lambda_{11,16} &&= \frac{3}{4}, \\
&\Lambda_{11,17} &&= \frac{1}{4}, \quad &&\Lambda_{12,21} &&= 1, \quad &&\Lambda_{15,16} &&= \frac{1}{4}, \quad &&\Lambda_{15,17} &&= 1, \\
&\Lambda_{15,18} &&= \frac{1}{2 \sqrt{2}}, \quad &&\Lambda_{16,17} &&= -\frac{1}{4}, \quad &&\Lambda_{17,18} &&= \frac{1}{2 \sqrt{2}}, \quad &&\Lambda_{19,43} &&= \frac{1}{2 \sqrt{2}}, \\
&\Lambda_{19,44} &&= -\frac{1}{2 \sqrt{6}}, \quad &&\Lambda_{19,45} &&= -\frac{1}{4\sqrt{3}}, \quad &&\Lambda_{19,46} &&= -\frac{\sqrt{5}}{4}, \quad &&\Lambda_{22,24} &&= -\frac{3}{2}, \\
&\Lambda_{25,27} &&= -\frac{1}{2}, \quad &&\Lambda_{28,31} &&= \frac{1}{4}, \quad &&\Lambda_{28,34} &&= -\frac{1}{4}, \quad &&\Lambda_{28,40} &&= -\frac{3}{4}, \\
&\Lambda_{29,41} &&= -2, \quad &&\Lambda_{30,31} &&= -\frac{1}{2\sqrt{2}}, \quad &&\Lambda_{30,34} &&= -\frac{1}{2 \sqrt{2}}, \quad &&\Lambda_{31,34} &&= \frac{3}{2}, \\
&\Lambda_{31,40} &&= \frac{1}{4}, \quad &&\Lambda_{32,37} &&= \frac{3}{2}, \quad &&\Lambda_{33,42} &&= -3, \quad &&\Lambda_{34,40} &&= -\frac{1}{4}, \\
&\Lambda_{35,36} &&= -\frac{1}{2\sqrt{2}}, \quad &&\Lambda_{35,38} &&= -\frac{1}{2 \sqrt{2}}, \quad &&\Lambda_{36,38} &&= \frac{7}{4}, \quad &&\Lambda_{43,44} &&= -\frac{\sqrt{3}}{4}, \\
&\Lambda_{43,45} &&= -\frac{1}{4}\sqrt{\frac{3}{2}}, \quad &&\Lambda_{43,46} &&= -\frac{3}{4}\sqrt{\frac{5}{2}}, \quad &&\Lambda_{44,45} &&= -\frac{1}{4\sqrt{2}}, \quad &&\Lambda_{44,46} &&= \frac{3}{4}\sqrt{\frac{3}{10}}, \\
&\Lambda_{44,47} &&= -\frac{3}{\sqrt{5}}, \quad &&\Lambda_{45,46} &&= \frac{1}{4}\sqrt{\frac{3}{5}}, \quad &&\Lambda_{45,47} &&= \frac{1}{2\sqrt{10}}, \quad &&\Lambda_{45,48} &&= -\sqrt{\frac{7}{2}}, \\
&\Lambda_{46,47} &&= \frac{11}{10 \sqrt{6}}, \quad &&\Lambda_{46,48} &&= \sqrt{\frac{7}{30}}, \quad &&\Lambda_{47,48} &&= \frac{1}{3}\sqrt{\frac{7}{5}}
\end{alignat*}}\relax
\begin{equation*}
\begin{alignedat}{4}
&L_i &&= 1,\quad &&\mathrel{i}&&=1,2,3,4,5,6,22,23,24,25,26,27, \\
&L_i &&= \frac{1}{\sqrt{2}},\quad &&\mathrel{i}&&=7,8,10,11,12,15,28,29,31,32,33,36,40,41,42, \\
&L_i &&= -\frac{1}{\sqrt{2}},\quad &&\mathrel{i}&&=13,16,17,19,20,21,34,37,38,
\end{alignedat}
\end{equation*}

\begin{equation*}
\begin{aligned}
&L_{43} &&= \frac{\sqrt{3}}{2}, &&L_{44} &&= \frac{1}{6} \left(1+2\sqrt{5}\right), \\
&L_{45} &&= \frac{5+\sqrt{5}+3 \sqrt{105}}{30\sqrt{2}}, &&L_{46} &&= \frac{1}{60} \left(-7 \sqrt{6}-3 \sqrt{14}+5\sqrt{30}\right), \\
&L_{47} &&= \frac{1}{30} \left(18-\sqrt{21}\right), &&L_{48} &&= \frac{1}{2}\sqrt{\frac{5}{3}}, \\
&M &&= L.
\end{aligned}
\end{equation*}


\bibliographystyle{JHEP}

\bibliography{ref}

\end{document}